\documentclass[a4paper,fleqn]{cas-sc}

\usepackage[authoryear,longnamesfirst]{natbib}

\newcommand{\bs}{\textbf{s}}
\newcommand{\bt}{\textbf{t}}

\def\tsc#1{\csdef{#1}{\textsc{\lowercase{#1}}\xspace}}
\tsc{WGM}
\tsc{QE}

\newtheorem{theorem}{Theorem}
\newtheorem{lemma}{Lemma}
\newproof{proof}{Proof}

\begin{document}
\let\WriteBookmarks\relax
\def\floatpagepagefraction{1}
\def\textpagefraction{.001}

\shorttitle{Degraded Point Processes}    

\shortauthors{K.M. Collins et~al.}  

\title [mode = title]{Analyzing spatial point processes degraded by displacement and imperfect detection}  



%

\author[1]{Kevin M. Collins}
\credit{Methodology, Software, Formal Analysis, Writing -- original draft}
\cormark[1]
\author[1]{Erin M. Schliep}
\credit{Supervision, Methodology, Writing -- review and editing}
\author[2]{Alan E. Gelfand}
\credit{Supervision, Methodology, Writing -- review and editing}
\author[3]{Tina M. Yack}
\credit{Data curation}
\author[4]{Christopher W. Clark}
\credit{Data curation}
\author[5]{Robert S. Schick}
\credit{Supervision, Methodology, Writing -- review and editing}




\affiliation[1]{organization={Department of Statistics, North Carolina State University},
            city={Raleigh},
            state={North Carolina},
            country={USA}}





\affiliation[2]{organization={Department of Statistical Science, Duke University},
            city={Durham},
            state={North Carolina},
            country={USA}}
            
\affiliation[3]{organization={Marine Geospatial Ecology Lab, Nicholas School of the Environment, Duke University},
            city={Durham},
            state={North Carolina},
            country={USA}
            }
\affiliation[4]{organization={K. Lisa Yang Center for Conservation Bioacoustics, Cornell University},
            city={Ithaca},
            state={New York},
            country={USA}
            }
\affiliation[5]{organization={Southall Environmental Associates, Inc.},
            country={USA}
            }

\cortext[1]{Corresponding author at: Department of Statistics, North Carolina State University. email: kmcolli9@ncsu.edu}



\begin{abstract}
    Spatial point processes are a valuable tool for probabilistic modeling to explain location data.  However, the data themselves are often observed imperfectly. In order to perform accurate inference, one must account for these imperfections, which we refer to as \textit{degradation}. We consider two forms of degradation for spatial Poisson processes: thinning and displacement. First, we provide some theoretical results on model identifiability, showing that, under weak conditions, one can jointly learn the scale of the displacement, a parametric form of thinning, and a nonparametric intensity function. The ability to learn all of these components and the resulting improvements for inference compared to the conceptual non-degraded but misspecified model are shown empirically via simulation study. Finally, we apply this approach to North Atlantic right whale call data from Cape Cod Bay.
\end{abstract}




\maketitle
\section{Introduction}

Spatial point processes are utilized to model random point configurations across disciplines such as ecology \citep{perry2006comparison, waagepetersen2016analysis, velazquez2016evaluation, ben2021spatial}, criminology \citep{shirota2017space}, and public health \citep{diggle2013spatial}. In practice, spatial point processes are often observed imperfectly, resulting in what we call \textit{degraded} point processes. See \cite{guttorp2026you} for a recent review on the impacts of different modes of degradation on various point process models, including intensity inference, parameter estimates, and popular summary statistics, such as Ripley's $K$-function. 

One of the most commonly used models is the Poisson process, which is specified via an intensity function $\lambda(\bs)$ where $\bs\in\mathcal{D}$, the domain over which the process is observed. Heuristically, if $\partial \bs$ denotes an arbitrarily small ball around $\bs$, then $\lambda(\bs)$ is defined as the limit of $P(N(\partial \bs) =1)/|\partial \bs|$ as the area $
|\partial \bs| \rightarrow 0$ with $N(A)$ the number of points in $A$.  We investigate inference for Poisson processes that have been degraded through two mechanisms: thinning (which for us corresponds to failed detection of a point) and displacement (which corresponds to measurement error in the observed location of the point). In order for the whole degraded model to be identifiable, assumptions must be made about the nature of the intensity and degradation.

Thinning a spatial point process can be thought of generatively: first, a point pattern is generated by the intensity $\lambda(\mathbf{s})$; then, given a point is present at $\bs$, it is observed with probability $p(\mathbf{s})\in [0,1]$. Here, $\lambda(\mathbf{s})$ is the non-degraded intensity that we seek to learn about and the thinning function, $p(\mathbf{s})$, which reflects the probability of a point at $\bs$ being observed. If $p(\bs)$ is constant for all $\bs\in\mathcal{D}$, then a uniform thinning is applied across the region. Otherwise, it can be specified as a function varying over $\bs$, perhaps as a function of covariates at that location. 
In the context of a real data application, we would distinguish between the two functions $\lambda(\bs)$ and $p(\bs)$ as reflecting spatial variation at the process-level and data level, respectively. For example, assume that a given spatial domain provides a suitable habitat for a particular species and the distribution of this species is defined according to a spatial process with intensity $\lambda(\mathbf{s})$. Spatial variation in this intensity could be due to environmental conditions, e.g., soil moisture or abundance of predator species, which alter the likelihood of species presence. The thinning process, $p(\mathbf{s})$, offers an additional stochastic modification that is independent of the primary process and captures the probability of actually observing these species. A common example in ecology is distance sampling, where an observer exists at a point or line transect, and their probability of observing an individual in the study population decreases with distance (i.e., as location of the individual gets further from the point/line transect) \citep{thomasDistanceSampling2013}. Simultaneous learning of distance sampling detection functions and intensity functions have previously been considered in \cite{johnson2010model} and \cite{yuan2017point}. As we discuss in Section 2, applying a thinning process to a Poisson process conveniently results in a new Poisson process of the observed point pattern with \emph{operating} intensity $p(\mathbf{s})\lambda(\mathbf{s})$. This general multiplicative form requires that we make some constraining assumptions to ensure identifiability.

The second form of degradation we consider is displacement, wherein the points are observed with measurement error relative to their true location. This can be written as
\begin{equation}
    \bs = \bs^* + \epsilon
\end{equation}
where $\bs$ is an observed location that has been displaced by $\epsilon\in\mathbb R^d$ from the true location $\bs^*$. This is immediately reminiscent of classical measurement error models \citep{carroll2006measurement}. In general, measurement error models are not identifiable since both $\bs^*$ and $\epsilon$ are specified stochastically. Typical assumptions to render this model identifiable are either a fully known distribution on $\epsilon$ or available auxiliary information, such as replicates without measurement error. However, there is a growing body of literature on nonparametric density estimation that loosens some of these assumptions by establishing identifiability when the density of interest and the error density have sufficiently different smoothness \citep{butucea2005minimax,meister2006density,apanasovich2021nonparametric}. A key result that allows us to leverage this literature is that a Poisson process with intensity $\lambda(\bs)$ whose points are independently displaced via some density $f(\bs)$ is once again a Poisson process with operating intensity $\int\lambda(\textbf{t})f(\textbf{s}-\textbf{t})d\textbf{t}$. This convolution is similar to what one would obtain when adding two random variables, as in a classical measurement error model. Thus, the problem of identifiability becomes one of \emph{blind deconvolution}.

Displacement in point processes has seen limited study thus far. Those that have studied it, such as \cite{lund2000models} and \cite{chakraborty2010analyzing}, either have included auxiliary information or assumed that the measurement error is known a priori, in order to conduct inference on the latent process. We show that this can be weakened, and that both the intensity function and distribution of the displacement can be learned simultaneously. In particular, we extend the univariate work of \cite{schwarz2010consistent} to identify both the density (intensity) function of $\bs^*$ and the covariance of multivariate Gaussian displacement in $d$-dimensional Euclidean space. This result hinges on $\bs^*$ being supported on a bounded subset of $\mathbb R^d$, which is a customary operating assumption for modeling point processes in order to ensure a finite point process \citep{moller2003statistical}. 

In what follows, we consider spatial Poisson processes that have been degraded under both thinning and displacement. We investigate the identifiability of this class of models by (i) allowing for a parametric thinning of the intensity and clarifying when this thinning is identifiable, and (ii) developing theoretical results in terms of deconvolution in $d$-dimensional Euclidean space. Then, we combine both sources of degradation to provide a rich Bayesian model formulation in which we can simultaneously learn about the parameters driving both degradations, as well as the latent intensity function modeled with a log-Gaussian process prior. Given only the degraded observations, this enables prediction of \textit{true} incidence in $\mathcal{D}$ or any arbitrary subregions within $\mathcal{D}$. The efficacy of this modeling approach is illustrated empirically through simulation studies.

Our real data example consists of a point pattern of vocalizing North Atlantic right whales (\textit{Eubalaena glacialis}) that were collected through a monitoring network of hydrophones in Cape Cod Bay, MA.  From a long deployment, we use data collected on 2009-02-19 from \citep{palmer2022accounting}. The locations of the calling whales are estimated by triangulation of calls that are detected on three or more hydrophones. This approach relies on the sound speed profile in the water column and the arrival times of the signal on the array of hydrophones \citep{watkinsSoundSourceLocation1972,spiesbergerPassiveLocalizationCalling1990}. It presents a challenging thinning process; because the localization procedure requires that an individual call be recorded on at least three hydrophones, we have to establish the probability of individual hydrophones detecting a call and then the composite probability. Additionally, the localization procedure itself is imperfect \citep{gruden2021tracking} and introduces displacement into the point pattern. Together, this thinning and displacement process results in a novel degradation of a spatial point process. 

The remainder of this paper is structured as follows. In section 2, we present the general degradation model and then consider each of the degradation processes independently and together to determine under what conditions the thinning, degradation, and latent intensity are simultaneously identifiable. This exploration includes a new result that establishes identifiability of the blind deconvolution of a compactly supported intensity measure and a multivariate Gaussian measure. In section 3, we discuss technical details of model fitting in the Bayesian setting including prior distributions, along with model assessment tools. Section 4 develops a useful simulation study to illustrate the efficacy of learning.  Section 5 provides the analysis of the acoustic whale data briefly described above. Section 6 concludes with a summary and discussion of future work.

\section{Methodology and theoretical results}

\subsection{A general hierarchical degradation model}

Let $\mathcal S^*=\{\bs_1^*,\dots,\bs_n^*\}$ be a latent spatial point pattern realized from a Poisson process supported on the compact domain $\mathcal{D}\subset\mathbb{R}^d$ with intensity function $\lambda(\bs)$. The distribution of $\mathcal S^*$ can be characterized by two properties:
\begin{enumerate}
    \item For a Borel set $B\subset\mathcal{D}$, the number of points in $B$, $N(B)\sim Poisson (\mu(B))$ where $\mu(B)=\int_B\lambda(\bs)d\bs$
    \item For disjoint Borel sets $B_1,\dots,B_k\subset \mathcal{D}$, $N(B_1),\dots,N(B_k)$ are independent
\end{enumerate}
Furthermore, the resulting likelihood of $\mathcal S^*$ is defined as
\begin{equation}
L(\lambda(\bs);\mathcal S^*)=\prod_{i=1}^{n}\lambda(\bs_i^*)\exp\{-\int_\mathcal{D}\lambda(\bs)d\bs\}
\end{equation}
where the unit rate Poisson process is the reference measure. 

In what follows, we consider the setting where we assume that $\mathcal S^*$ is not observed directly. Instead, what we observe is a degraded realization of the spatial point pattern denoted $\mathcal S$. Given the observed $\mathcal S$, our goal is still to conduct inference on the latent $\lambda(\bs)$. The degradation processes that we explore are those discussed in \cite{lund2000models} through the following transformations:
\begin{enumerate}
    \item \textit{Thinning}: Each point $\bs_i^*\in\mathcal S^*$ is retained with probability $p(\bs_i^*)$, where $\mathcal S$ defines the resulting collection of retained points. The probability of retaining a point is independent of any other point.
    \item \textit{Displacement}: Each point $\bs_i^*\in\mathcal S^*$ is displaced to a new position $\bs_i\in\mathcal S$ with probability density function $f(\bs_i;\bs_i^*,\phi)$ where $\phi$ is the displacement parameter. 
\end{enumerate}

Note that \cite{lund2000models} also include a third and fourth transformation, referred to as ``superposition" and ``censoring,'' respectively. The superposition of ``ghost points" can be interpreted as the presence of false positives in the data.  The study of ghost points is deferred to future work as identifiability of a ghost intensity is not easy to establish in the context of Poisson processes. Censoring is a consequence of a point being displaced outside of the region of observation, and thus being missed, which is considered implicitly in our exploration of displacement. As in \cite{lund2000models} and \cite{chakraborty2010analyzing}, we adopt an ``island" model where points may be censored by being displaced outside of $\mathcal{D}$ (and thus unobserved) but there are no points in $\mathcal{D}^c = \mathbb R^d\backslash\mathcal{D}$ that may be displaced into our region of observation. 

Conveniently, when a Poisson process undergoes any of the above transformations, it remains a Poisson process but with a modified intensity function \citep{haenggi2013stochastic}. 
\begin{theorem}
    Let $\mathcal S^*$ be a stationary Poisson process with intensity function $\lambda(\bs)$.
    \begin{itemize}
        \item If $\mathcal S$ is the set of points independently retained with probability $p(\bs)$, then $\mathcal S$ is a Poisson process with intensity function $p(\bs)\lambda(\bs)$.
        \item If each $\bs_i^*$ is independently displaced to $\bs_i$ with probability density $f(\bs_i;\bs_i^*,\phi)$, then $\mathcal S$ is a Poisson process with intensity function $\lambda_f(\bs)=\int_{\mathbb{R}^d}\lambda(\bt)f(\bs;\bt,\phi)d\bt$
    \end{itemize}
\end{theorem}
For the second result, if the displacement is based on the difference $\bs-\bt$, then it can be expressed as the convolution $\lambda_f(\bs)=\int_{\mathbb R^d}\lambda(\bt)f(\bs-\bt;\phi)d\bt$, written alternatively as $\lambda_f(\bs)=(\lambda*f)(\bs)$.

We can specify the model in multiple ways. First, we present a general hierarchical model for a degraded Poisson process as
\begin{equation}
    \mathcal S^*\sim PP(\lambda(\bs))
\end{equation}
\begin{equation}
    \mathcal S^\dagger = \{\bs_i^*: \gamma_i=1\}, \gamma_i\sim Bern(p(\bs_i^*))
\end{equation}
\begin{equation}
    \mathcal S=\{\bs_i\}, \bs_i \sim f(\bs_i;\bs_i^\dagger,\phi)
\end{equation}
where $\mathcal{S}$ is the displacement of $\mathcal{S}^\dagger$, which in turn is the thinned realization of $\mathcal{S}^*$. Although this hierarchical formulation elegantly describes the data generating process and naturally fits into a Bayesian framework, the sampling of latent $\gamma_i$ and $\bs_i^*$ may be computationally burdensome due to the increased parameter dimension, particularly when our main interest is the latent intensity $\lambda(\bs)$. Alternatively, we can use the aforementioned results for an equivalent, but more direct formulation by defining
\begin{equation}
    \mathcal{S}\sim PP(\nu(\bs))
\end{equation}
where
\begin{equation}
    \nu(\bs)=\int_{\mathbb R^d}p(\bt)\lambda(\bt)f(\bs-\bt;\phi)d\bt.
\end{equation}
For convenience, we sometimes write this as $\nu(\bs)=(f*p\lambda)(\bs)$. Note that we assume that the thinning of points occurs before displacement. One could assume the reverse, which would result in
\begin{equation}
    \nu(\bs) = p(\bs)\int_{\mathbb R^d}\lambda(\bt)f(\bs-\bt;\phi)d\bt
\end{equation}
Both are valid, but identifiability may be more challenging to assume under the latter, as we will see in the next section. Ultimately, the order in which degradation occurs will be dependent upon the application, but it is important to acknowledge that these two orderings do not result in the same intensity, $\nu(\bs)$.

\subsection{Identifiability of degradation components}

We now illuminate conditions under which the model in (7) is identifiable. In what follows, we only consider Poisson processes that are compactly supported on some set $\mathcal{D}\subset\mathbb R^d$. Let $\Lambda(A)$ denote the intensity measure of a Poisson process, where $\Lambda(A)=\int_A\lambda(\bs)ds$ and $\lambda(\bs)$ denotes the intensity function. First, we note that a Poisson process is uniquely identified by its intensity measure. That is, if $\Lambda_1(A)=\Lambda_2(A)$ for all Borel sets $A\subset\mathcal{D}$, then the Poisson processes resulting from each of these intensities are identical \citep{baddeley2016spatial}. An equivalent definition is that the intensity functions $\lambda_1(\bs)=\lambda_2(\bs)$ almost everywhere. Thus, in order to show that a Poisson process model indexed by a (possibly infinite dimensional) parameter $\theta$ is identifiable, it suffices to show that $\lambda(\bs;\theta_1)=\lambda(\bs;\theta_2)$ almost everywhere implies $\theta_1=\theta_2$. We now consider each degradation component separately to study identifiability of the model.

\subsubsection{Thinning}

First, assume that $\mathcal{S}$ has a latent intensity $\lambda(\bs)$ and is degraded by some thinning function $p(\bs)\in[0,1]$, resulting in an operating intensity function $\nu(\bs)=p(\bs)\lambda(\bs)$. The functions $\lambda(\bs)$ and $p(\bs)$ are identifiable if  $\log p(\bs)$ and $\log \lambda(\bs)$ are linearly independent almost everywhere; that is, $a\log p(\bs) + b\log\lambda(\bs) = 0$ if and only if $a=b=0$ for almost all $\bs\in\mathcal{D}$. Linear independence of $\lambda(\bs)$ and $p(\bs)$ requires that only one of the functions can have a log ``intercept" term. In practice, we typically assume that the intercept term belongs to the intensity function and not the thinning function in order to capture the global mean of the intensity function across the region. Additionally, it is often assumed that detection is perfect at some $\bs\in\mathcal{D}$, so $p(\bs)=1$ there, prohibiting a log intercept.

To see an example of this, we consider a parametric form of $p(\bs)$ that is typical in ecological modeling where data are collected through distance sampling \citep{bucklandDistanceSampling1993}.  Namely, sampling occurs along a transect and detection decays as a function of distance from the transect. For example, letting $d(\bs)$ denote the distance between location $\bs$ and the transect, a general form of decay function often used is
\begin{equation}
    p(\bs) = p_0\exp\{-d(\bs)^{1/\gamma}/\eta\}
\end{equation} 
where $p_0\in(0,1]$, $\eta>0$, and $\gamma>0$. When $\gamma=1$, this simplifies to the exponential decay function. Under this specification, the log of the operating intensity is
\begin{equation}
    \log\nu(\bs) = \log p_0 - d(\bs)/\eta + \log\lambda(\bs)
\end{equation}
Now, the identifiability of $p_0$ and $\eta$ depends upon the form of $\log\lambda(\bs)$. As we stated previously, only one of the thinning or intensity functions can have a log intercept term, so we assume that $p_0$ is fixed and known. In distance sampling, it is often assumed that detection at the transect (i.e., when  $d(\bs)=0$) occurs with probability 1, meaning $p_0=1$.

Assuming a nonhomogenous Poisson process (NHPP), we define  $\log\lambda(\bs)=X(\bs)'\boldsymbol\beta$ for some known covariates $X(\bs)$. If $X(\bs)$ and $d(\bs)$ are linearly independent, then $\eta$ and $\boldsymbol\beta$ are identifiable. For a more flexible model, we could assume that $\log \lambda(\bs)$ is a realization from a random process, such as a Gaussian process, resulting in the popular log-Gaussian Cox process. Once again, if $\log\lambda(\bs)$ and $p(\bs)$ are linearly independent, then each are identifiable, which will occur with probability $1$ in the case of $\log\lambda(\bs)$ being a random realization of a Gaussian process.

\subsubsection{Displacement}

Next we turn to point process degradation on a bounded set $\mathcal{D} \subset R^{d}$, arising from  displacement. Recall that a Poisson process is uniquely defined by its intensity measure, which is a finite measure. Let $\mu$ denote a finite measure and $\mathcal N_\Sigma$ denote a multivariate Gaussian measure with mean $\mathbf{0}$ and covariance $\Sigma$.

Theorem 1 stated above shows that the intensity of a displaced Poisson process is the convolution of the original intensity function with the displacement density. In what follows, we speak somewhat more generally in terms of measures and we use the notation $\mu*\nu$ to denote the convolution of two measures $\mu$ and $\nu$, that is
\begin{equation}
    (\mu*\nu) (A) = \int_{\mathbb R^d}\int_{\mathbb R^d} \textbf{1}(x+y\in A) d\mu(x) d\nu(y).
\end{equation}

Here, we show that any Poisson process that has a compactly supported intensity measure and is displaced by a Gaussian error term is identifiable. This result follows a one-dimensional result, Theorem 2.1, developed in \cite{schwarz2010consistent}.  We show by contradiction that the deconvolution of $\mu * \mathcal{N}_\Sigma$ is uniquely identifiable through an argument of the sets on which these convolutions are supported. First, we use an established lemma regarding the support of a convolution of measures.

\begin{lemma}
    If $\mu$ is a finite measure with compact support and $\nu$ is a finite measure with closed support, then
    \begin{equation*}
        \text{supp}\{\mu*\nu\}=\{a+b: a\in\text{supp}\{\mu\}, b\in\text{supp}\{\nu\}\}
    \end{equation*}
\end{lemma}
\begin{proof}
    If scaled by their total measure, $\mu$ and $\nu$ can both be treated as probability measures, $\bar \mu$ and $\bar \nu$. As such, the interpretation of a convolution as the sum of two random variables $X\sim\bar\mu$ and $Y\sim\bar\nu$ leads directly to the support of $\mu*\nu$ being the closure of the sumset of supp$\{\mu\}$ and supp$\{\nu\}$. The sumset of a compact set and a closed set is necessarily closed, so the support is just the sumset described.
\end{proof}

\begin{theorem}
    Let $\mu_1$ and $\mu_2$ be two finite measures with full support on the compact set $\mathcal{D}\subset \mathbb R^d$, and $\mathcal{N}_{\Sigma_1}$, $\mathcal{N}_{\Sigma_2}$ be Normal measures with mean 0 and covariance matrices $\Sigma_1$ and $\Sigma_2$ respectively. Then, $\mu_1*\mathcal{N}_{\Sigma_1}=\mu_2*\mathcal{N}_{\Sigma_2} \implies \mu_1=\mu_2$ and $\Sigma_1=\Sigma_2$. 
\end{theorem}

\begin{proof}
    Assume for contradiction that $\Sigma_2\neq\Sigma_1$. By the convolution theorem, we can rewrite the assumed convolution equality under a Fourier transform as
    \begin{equation}
        \hat\mu_1(\xi) = \hat\mu_2(\xi) \exp\{-\frac{1}{2}\xi'(\Sigma_2-\Sigma_1)\xi\}
    \end{equation}
    The two possible scenarios to consider are $\Sigma_2-\Sigma_1$ being either definite or indefinite. If definite, then we proceed without loss of generality by assuming that it is positive definite and note that this statement implies (again by the convolution theorem) that
    \begin{equation}
        \mu_1=\mu_2*\mathcal{N}_{\Sigma_2-\Sigma_1}.
    \end{equation}
    However, by Lemma 1, supp$\{\mu_2*\mathcal{N}_{\Sigma_2-\Sigma_1}\}=\mathbb R^d$, while $\mu_1$ is compactly supported, which is a contradiction. This is only rectified if $\Sigma_2=\Sigma_1$.
    
   We now proceed to the indefinite case. The matrix $\Sigma_2-\Sigma_1$ is symmetric and thus can be decomposed into $\Sigma_2-\Sigma_1=VDV$ where $D$ is a diagonal matrix such that $D_{ii}=\lambda_i$, the $i$th eigenvalue of $\Sigma_2-\Sigma_1$. Define two new diagonal matrices $B$ and $C$ such that
    \begin{equation}
        B_{ii}=\begin{cases}
            |\lambda_i| \text{ if } \lambda_i\neq 0 \\
            1 \text{ if } \lambda_i=0
        \end{cases}
    \end{equation}
    and
    \begin{equation}
        C_{ii}=\begin{cases}
            |\lambda_i| - \lambda_i \text{ if } \lambda_i\neq 0 \\
            1 \text{ if } \lambda_i=0.
        \end{cases}
    \end{equation}
    It is clear then that $D=B-C$, and so $\Sigma_2-\Sigma_1=\Omega_2-\Omega_1$ where $\Omega_2=VBV$ is positive definite and $\Omega_1=VCV$ is positive semidefinite. Note that $\Omega_1$ is positive definite if and only if $\lambda_i=0$ for all $i$. Reordering again
    \begin{equation}
        \hat\mu_1(\xi)\exp\{-\frac{1}{2}\xi'(\Omega_1)\xi\} = \hat\mu_2(\xi) \exp\{-\frac{1}{2}\xi'(\Omega_2)\xi\}
    \end{equation}
    and we now have 
    \begin{equation}
        \mu_1*\mathcal{N}_{\Omega_1}=\mu_2*\mathcal{N}_{\Omega_2}
    \end{equation}
    where $\mathcal{N}_{\Omega_1}$ is a degenerate multivariate Normal measure and $\mathcal{N}_{\Omega_2}$ is a \textit{non}degenerate multivariate Normal measure.

    Again by Lemma 1, $\text{supp}\{\mu_2*\mathcal{N}_{\Omega_2}\}=\mathbb R^d$ and $\text{supp}\{\mu_1*\mathcal{N}_{\Omega_1}\}=\{a + \Omega_{1}^{1/2}b: a\in\mathcal{D}, b\in \mathbb R^d\}$. If $\Omega_1$ has rank $m<d$, then $\{\Omega_1^{1/2}b:b\in\mathbb R^d\}$ is necessarily an $m$-dimensional subspace of $\mathbb R^d$. Let $\{v_1,\dots,v_m\}$ denote a set of basis vectors for $\{\Omega_1^{1/2}b:b\in\mathbb R^d\}$, which are a subset of the basis for $\mathbb{R}^d$, $\{v_1,\dots,v_m,v_{m+1},\dots,v_d\}$. Now, let $x_0\in\text{span}\{v_{m+1},\dots,v_d\}$. Assume that there exists a collection of vectors $\{w\in\mathcal{D}:x_0-w\in\text{span}\{v_1,\dots,v_m\}\}$. Note that because $\mathcal{D}$ is compact (bounded), there exists $r>0$ such that $rw\notin\mathcal{D}$ for any $w\in\mathcal{D}$. If we consider the vector $rx_0$, it becomes evident that $rx_0-w$ can be decomposed into the linear combination of a vector in span$\{v_1,\dots,v_m\}$ and a vector in span$\{v_{m+1},\dots,v_d\}$. Therefore, $rx_0\notin\text{supp}\{\mu_1*\mathcal{N}_{\Omega_1}\}\implies \text{supp}\{\mu_1*\mathcal{N}_{\Omega_1}\}\neq\text{supp}\{\mu_2*\mathcal{N}_{\Omega_2}\}$, arriving at a contradiction. In fact, $\text{supp}\{\mu_1*\mathcal{N}_{\Omega_1}\}=\mathbb R^d$ if and only if $\Omega_1$ is positive definite, which only occurs if $\lambda_i=0$ for all $i$, i.e. if $\Sigma_2=\Sigma_1$. From this, if $\Sigma_1=\Sigma_2$, then $\mu_1=\mu_2$, which concludes the proof.
\end{proof}

\subsubsection{Combining thinning and displacement}

We now address identifiability when both thinning and displacement are present. First, we note that if a Poisson process undergoes thinning, yielding a new Poisson process with intensity function
\begin{equation}
    \psi(\bs)=p(\bs)\lambda(\bs)
\end{equation}
and it is uniquely identified with respect to $p(\bs)$ and $\lambda(\bs)$, then it follows immediately from Theorem 2 that each of these terms remains identifiable in addition to the displacement covariance. To see this, let $\Psi(A)=\int_A\psi(\bs)d\bs$, which, in turn, means that  $\Psi_1*\mathcal{N_{\sigma_1}}=\Psi_2*\mathcal{N_{\sigma_2}}\implies \Psi_1=\Psi_2$ and $\Sigma_1=\Sigma_2$. Conveniently, to establish identifiability we need only $\Psi$ to be compactly supported; that is, if $\lambda(s)$ is nonzero outside the region of observation, but $p(s)$ is compactly supported, then identifiability of displacement still holds. This assumption may be more comfortable for practitioners, particularly as one could interpret $p(s)$ (detection) necessarily being zero outside the region of observation. On the other hand, if we first apply displacement, such that $\psi=\lambda*\phi$, then it may be more challenging to justify the identifiability of $p(\bs)$ in $\nu(\bs)=p(\bs)\psi(\bs)$. If we look at the simpler case of an NHPP based on covariates as done previously, the log of the operating intensity is
\begin{equation}
    \log\nu(\bs)=-d(\bs)/\eta + \log\int_{\mathcal D} \phi_\Sigma(\bs-\bt) \exp\{X(\bt)'\boldsymbol\beta\}d\bt
\end{equation}
Earlier, we established identifiability based on linear independence of $d(\bs)$ and $X(\bs)$. In this case, $\nu(\bs)$ is no longer even log-linear in $\boldsymbol\beta$. Given this insight and the more broadly appropriate interpretation of thinning-then-displacing, we only handle operating intensities of the form $\nu(\bs)=(\phi_\Sigma*p\lambda)(\bs)$ for the remainder of the paper.

\section{Bayesian model fitting and model assessment}

For the simulation studies and real data analysis that follows, we adopt the following model. Let $\mathcal{S}=\{\bs_1,\dots,\bs_m\}$ denote a degraded point pattern observed on the compact domain $\mathcal{D}\subset \mathbb R^d$ with latent intensity function $\lambda(\bs;\boldsymbol\beta)$ that has first been thinned by a detection function $p(\bs;\boldsymbol \eta)$, indexed by a finite-dimensional parameter $\boldsymbol\eta$, and then displaced by a $d$-dimensional Gaussian density with mean $\textbf{0}$ and covariance $\Sigma$. Finally, we are interested in a flexible specification for the intensity $\lambda(\bs)$. In the absence of covariates, an illustrative way to achieve this is by specifying $\log \lambda(\bs;\boldsymbol\beta)$ as a linear function of a finite collection of orthogonal basis functions $\psi_k(\bs)$ for $k=1,\dots,K$:
\begin{equation}
    \log\lambda(s;\boldsymbol\beta)= \beta_0 + \sum_{k=1}^K\psi_k(\bs)\beta_k.
\end{equation}
Let $\mathcal{S}$ denote a realization from a Poisson process with operating intensity function $\nu(\bs)=(\phi_\Sigma*p\lambda)(\bs)$. The resulting log likelihood function is written
\begin{equation}
\begin{aligned}
        \ell(\boldsymbol\beta,\eta,\Sigma;\mathcal{S})=  \sum_{i=1}^m\log \int_{\mathcal D} \phi_\Sigma(\bs_i-\bt)p(\bt;\eta)\lambda(\bt;\boldsymbol\beta)d\bt- \\{\int_{\mathcal{D}}\int_{\mathcal D} \phi_\Sigma(\bs-\bt)p(\bt;\eta)\lambda(\bt;\boldsymbol\beta)d\bt d\bs}.
\end{aligned}
\end{equation}
Evaluating this log likelihood function requires computing two integrals, the displacement convolution term and the compensator term. Given that neither of these are available analytically in general, we approximate both using quadrature. The domain $\mathcal{D}$ is discretized into $J$ grid cells where $\textbf{d}_j$ denotes the centroid and $|\textbf{d}_j|$ denotes the area of grid cell $j$ for $j=1,\dots,J$. For the convolution, we use a finite sum approximation
\begin{equation}
    \nu(\bs;\boldsymbol\beta,\eta,\Sigma)=\int_{\mathbb R^d}p(\bt;\eta)\lambda(\bt)\phi_{\Sigma}(\bs-\bt)d\bt\approx\sum_{j=1}^J p(\textbf{d}_j;\eta)\lambda(\textbf{d}_j;\boldsymbol\beta)\phi_{\Sigma}(\bs-\textbf{d}_j)|\textbf{d}_j|.
\end{equation}
Next, let $n_j$ denote the number of points in $\mathcal{S}$ that fall within the grid cell centered at $\textbf{d}_j$, such that $\sum_{j=1}^J n_j=m$. We approximate the likelihood function by assuming $\nu(\bs;\boldsymbol\beta,\eta\Sigma)$ is constant within each grid-cell such that the discretized log likelihood can be written as
\begin{equation}
    \ell(\boldsymbol\beta,\eta,\Sigma;\mathcal{S}) = \sum_{j=1}^J n_j\log\nu(\textbf{d}_j;\boldsymbol\beta,\eta,\Sigma)-\nu(\textbf{d}_j;\boldsymbol\beta,\eta,\Sigma)|\textbf{d}_j|.
\end{equation}

To complete the hierarchical Bayesian formulation, we assign prior distributions to all model parameters. For the basis coefficients, $\beta_1,\dots,\beta_K$ we assume independent $N(0,\tau^2)$ distributions. For the remaining parameters, if lacking substantial prior information, we assign uninformative and conjugate priors when available. Thus, the parameters being estimated in this model are $\{\boldsymbol \beta,\tau^2,\eta,\Sigma\}$ and the posterior of interest has the form
\begin{equation}
    \pi(\boldsymbol \beta,\tau^2,\eta,\Sigma|\mathcal{S})\propto L(\boldsymbol \beta,\eta,\Sigma;\mathcal{S})\pi(\boldsymbol\beta|\tau^2)\pi(\tau^2)\pi(\eta)\pi(\Sigma)
\end{equation}
In order to sample from this posterior we employ a straightforward Metropolis-within-Gibbs MCMC algorithm. 

In terms of model assessment, we are concerned with two inferential goals: estimation of the degradation parameters $\eta$ and $\Sigma$ and prediction of the latent $\lambda(\bs)$. When direct parameter comparison is possible, we consider the point estimates of $\eta$ and $\Sigma$ and their corresponding empirical coverage rates. Given that each of these quantities is positively supported and often results in a heavy-tailed posterior, we report posterior medians and credible intervals for parameter estimation. 
For $\lambda(\bs)$, we evaluate inference at a given $\bs$ via mean absolute deviation (MAD) of the posterior median and empirical coverage of credible intervals. Furthermore, following notions of residuals developed in \cite{baddeley2008properties} and expanded upon in the Bayesian setting in \cite{leininger2017bayesian}, we consider inference on the expected number of counts in a (sub)region $A\subseteq\mathcal{D}$. To do this, because we cannot necessarily compare parameters directly across all models, we consider posterior samples of $\int_A\lambda(\bs)d\bs$. This quantity reflects the expected number of points \emph{before} the point pattern has been degraded by thinning and displacement, which is of great interest in disciplines such as ecology as this gives rise to the notion of ``true" unobserved abundance. As before, we consider MAD and empirical coverage to evaluate the ability of the model to recover this quantity.

\section{Simulation study}

We now present several simulation studies to empirically assess the performance of these models in estimating thinning and displacement parameters and the latent intensity, and to compare model performance against insufficiently complex models. Without loss of generality, we restrict ourselves to observations on the unit square $\mathcal{D}=[0,1]^2$. For thinning, we assume that detection probability decreases exponentially in perpendicular distance from a vertical line transect in the middle of the region. Specifically, we define $\log p(\bs)=-||0.5-s_1||/\eta$ where $\eta$ is unknown and to be estimated and $s_1$ is the $x$-coordinate of $\bs$. For displacement, we assume $\phi(\bt-\bs;\Sigma)$ is a bivariate Gaussian density function where $\Sigma$ is diagonal with unique entries $\Sigma_{11}\neq\Sigma_{22}$. Finally, the basis function $\psi_k(\bs)$ used to construct the log of the true intensity are the first 100 eigenvectors of the correlation matrix that is the result of an exponential covariance function with fixed range parameter $\rho=0.1$ evaluated using the grid cell centroids. With respect to prior distributions, for $\beta_0$ we use an uninformative Normal mean 0 prior with large variance. We then assign an uninformative inverse Gamma prior on $\tau^2$ for conjugacy. For the diagonal components of $\Sigma$, we adopt an uninformative Gamma prior. Due to the collinearity of $d(\bs)$ and some of the lower-frequency $\psi_k(\bs)$, $\eta$ is given a Gamma prior such that its mean is centered at the true value and it has a variance of 0.1. We believe this informativeness is justified given that an analyst would have a reasonable sense of how detection decays a priori under distance sampling.

We consider a number of different scenarios to understand when and how well we are able to recover these parameters. First, we modify the expected number of observed points after thinning, $E[N_{obs}(\mathcal{D})]=\{300,500,1000\}$, to see how the number of observed points impacts performance. We additionally consider combinations of three parameters: 1) $\tau^2=\{1,3\}$ to see how volatility in the intensity changes inference; 2) $(\Sigma_{11},\Sigma_{22})=\{(0.005,0.0025),(0.01,0.005)\}$ to explore varying levels of displacement variance; and 3) $\eta=\{0.3,0.5\}$ to explore different rates of decay in detection. Six possible combinations of these values are detailed below, with additional results presented in the Appendix. For each scenario, we generate 100 simulated point patterns, each from a unique realization of the Gaussian process. The models are fit using an MCMC sampler that runs for 60,000 iterations, discarding the first 30,000 samples as burn-in.

Under each scenario, we fit four different models (let $\nu(\bs)$ denote the operating intensity): (M1) $\nu(\bs)=\lambda(\bs)$, (M2) $\nu(s)=p(\bs)\lambda(\bs)$, (M3) $\nu(\bs)=(\phi*\lambda)(\bs)$, and (M4) $\nu(\bs)=(\phi*p\lambda)(\bs)$. Note that M4 is the true generating process wherein we have both displacement and thinning as degradations of $\lambda(\bs)$. This will provide perspective on how incorporating both of these degradations is important for inference on the true intensity.

Table \ref{lambda_combined} provides inference metrics for the intensity $\lambda(\bs)$ under each model. Outside of the first two cases where the number of observed points is small and the intensity has low-variance, M4 outperforms the alternatives in MAD. For $E[N_{obs}(\mathcal{D})]=300$, this indicates that fewer observed points results in less reliable estimates all of model pieces simultaneously, which is why M1 performs the best. When $\tau^2=1$, M4 also performs worse than M2, suggesting that a relatively flat intensity surface makes it more challenging to estimate the displacement. Thus, the resulting deconvolved intensity function is missestimated, whereas the thinning function is still critical. In terms of uncertainty quantification, Table \ref{lambda_combined} shows that the 90\% credible intervals produced by M4 generally achieve nominal coverage rates empirically, while the other three models fall short. M1-M3 are all misspecified so we would not expect them to achieve nominal rates, but the dramatic improvement of M4 indicates the importance of considering both degradation pieces simultaneously.

\begin{table}
\caption{Average pixel-wise $\lambda(\bs)$ MAD of posterior median and 90\% credible interval empirical coverage rates under models M1-4.}
\centering
\begin{tabular}{cccccccccccc}
\hline
 &  &  &  & \multicolumn{4}{c}{MAD} & \multicolumn{4}{c}{Coverage} \\
 \cmidrule(lr){5-8}
\cmidrule(lr){9-12}
\hline
$E[N_{obs}(\mathcal{D})]$ & $\tau^2$ & $\Sigma$ & $\eta$ 
& M1 & M2 & M3 & M4 
& M1 & M2 & M3 & M4 \\
\hline
300 & 3 & 1 & 0.3 
& \textbf{467} & 486 & 490 & 524 
& 0.46 & 0.49 & 0.74 & \textbf{0.92} \\
\hline
500 & 1 & 1 & 0.3 
& 648 & \textbf{531} & 730 & 626 
& 0.36 & 0.58 & 0.70 & \textbf{0.92} \\
\hline
500 & 3 & 0 & 0.3 
& 741 & 707 & 735 & \textbf{624} 
& 0.44 & 0.55 & 0.68 & \textbf{0.90} \\
\hline
500 & 3 & 1 & 0.3 
& 782 & 822 & 788 & \textbf{728} 
& 0.41 & 0.45 & 0.71 & \textbf{0.88} \\
\hline
500 & 3 & 1 & 0.5 
& 567 & 589 & 556 & \textbf{517} 
& 0.44 & 0.45 & 0.81 & \textbf{0.90} \\
\hline
1000 & 3 & 1 & 0.3 
& 1599 & 1613 & 1565 & \textbf{1469} 
& 0.33 & 0.39 & 0.67 & \textbf{0.85} \\
\hline
\end{tabular}
\label{lambda_combined}
\end{table}

Inference about $N(\mathcal{D})$, the ``true" number of points before thinning and displacement, is reported in Table \ref{EN_combined}. A point estimate is obtained for this quantity via the posterior median of the intensity function integrated over the domain $\int_{\mathcal{D}}\lambda(\bs)d\bs$. Note that this quantity will necessarily be larger than the observed number of points. We report the posterior median under each model and compare models via the MAD. Models M2 and M4 outperform M1 and M3 with similar results as to $\lambda(\bs)$. This reemphasizes the importance of correctly estimating thinning when trying to estimate true counts. M4 has the best empirical coverage, though we find that under some scenarios the 90\% credible interval is too conservative which could be an artifact of the posterior distributions being heavy-tailed.

\begin{table}
\caption{Average $N(\mathcal{D})$ MAD of posterior median and 90\% credible interval empirical coverage rates under models M1-4.}
\centering
\begin{tabular}{cccccccccccc}
\hline
 &  &  &  & \multicolumn{4}{c}{MAD} & \multicolumn{4}{c}{Coverage} \\
 \cmidrule(lr){5-8}
\cmidrule(lr){9-12}
\hline
$E[N_{obs}(\mathcal{D})]$ & $\tau^2$ & $\Sigma$ & $\eta$ 
& M1 & M2 & M3 & M4 
& M1 & M2 & M3 & M4 \\
\hline
300 & 3 & 1 & 0.3 
& 352 & \textbf{173} & 303 & 566 
& 0.00 & 0.77 & 0.02 & \textbf{0.97} \\
\hline
500 & 1 & 1 & 0.3 
& 577 & \textbf{218} & 464 & 386 
& 0.00 & 0.71 & 0.06 & \textbf{0.89} \\
\hline
500 & 3 & 0 & 0.3 
& 601 & 316 & 553 & \textbf{261} 
& 0.00 & 0.74 & 0.00 & \textbf{0.94} \\
\hline
500 & 3 & 1 & 0.3 
& 600 & 363 & 536 & \textbf{340} 
& 0.00 & 0.59 & 0.01 & \textbf{0.92} \\
\hline
500 & 3 & 1 & 0.5 
& 357 & 163 & 276 & \textbf{97} 
& 0.00 & 0.79 & 0.02 & \textbf{1.00} \\
\hline
1000 & 3 & 1 & 0.3 
& 1279 & \textbf{721} & 1159 & 840 
& 0.00 & 0.55 & 0.00 & \textbf{0.87} \\
\hline
\end{tabular}
\label{EN_combined}
\end{table}

In terms of estimating the thinning process, we report posterior medians, standard deviations, and empirical coverages of $\eta$. Note that $\eta$ is only estimated under M2 and M4. Table \ref{eta} shows the posterior median is closer to the truth with smaller standard deviations under M4 in all scenarios, but the two models are sufficiently close that we consider their performance roughly equivalent. However, in terms of coverage, M4 is consistently closer to the nominal coverage rate with the exception of the scenario with $\eta=0.5$, where it is too conservative and M2 achieves the nominal coverage rate.

\begin{table}
\caption{$\eta$ mean (SD) of posterior median and 90\% credible interval empirical coverage rates under models M2 and M4.}
\centering
\begin{tabular}{cccccccc}
\hline
 &  &  &  & \multicolumn{2}{c}{Mean (SD)} & \multicolumn{2}{c}{Coverage} \\
 \cmidrule(lr){5-6}
\cmidrule(lr){7-8}
\hline
$E[N_{obs}(\mathcal{D})]$ & $\tau^2$ & $\Sigma$ & $\eta$ & M2 & M4 & M2 & M4 \\
\hline
300 & 3 & 1 & 0.3 & 0.35 (0.15) & \textbf{0.31} (0.13) & 0.81 & \textbf{0.97} \\
\hline
500 & 1 & 1 & 0.3 & 0.34 (0.12) & \textbf{0.30} (0.11) & 0.78 & \textbf{0.92} \\
\hline
500 & 3 & 0 & 0.3 & 0.36 (0.16) & \textbf{0.35} (0.12) & 0.78 & \textbf{0.95} \\
\hline
500 & 3 & 1 & 0.3 & 0.36 (0.15) & \textbf{0.35} (0.13) & 0.72 & \textbf{0.94} \\
\hline
500 & 3 & 1 & 0.5 & 0.51 (0.18) & \textbf{0.50} (0.11) & \textbf{0.91} & 1.00 \\
\hline
1000 & 3 & 1 & 0.3 & 0.40 (0.20) & \textbf{0.38} (0.16) & 0.62 & \textbf{0.87} \\
\hline
\end{tabular}
\label{eta}
\end{table}

For the displacement covariance, Table \ref{Sigma_combined} illustrates that both M3 and M4 have comparable posterior medians within each simulation setting, each overestimating the truth in general. It is interesting to note that the worst performance occurs when $\tau^2=1$, indicating that  displacement is more challenging to estimate when the latent intensity is less variable. Comparing models across $E[N_{obs}(\mathcal{D})]$, processes with higher overall intensity (more points) seem to significantly increases  the ability to recover the true value. Finally, Table \ref{Sigma_cov} provides empirical coverage rates for the parameters $\Sigma_{11}$ and $\Sigma_{22}$. We again find that M4 has consistently better coverage than M3, but the credible intervals are too narrow on average under each scenario. However, we see that empirical coverage tends to improve as more points are observed. We have limited ourselves to two data generating values of $\Sigma$, but note that smaller variances, while identifiable, may be difficult to estimate in practice and more sensitive to the granularity of the discretization scheme.

The results from these simulation studies illustrate both the efficacy of our modeling approach and the importance of accounting for degradation. In particular, we see how failing to account for degradation leads to bias and poor uncertainty quantification. In addition, we are able to estimate the parameters governing degradation, which, in practice, will allow us to better understand the limitations of collected data.

\begin{table}
\caption{$\Sigma_{11}$ and $\Sigma_{22}$ mean (SD) of posterior medians under models M3 and M4. For $\Sigma$: 0 corresponds to $(\Sigma_{11},\Sigma_{22}=(0.005,0.0025)$, and 1 corresponds to $(\Sigma_{11},\Sigma_{22})=(0.010,0.005)$.}
\centering
\begin{tabular}{cccccccc}
\hline
 &  &  &  & \multicolumn{2}{c}{$\Sigma_{11}$} & \multicolumn{2}{c}{$\Sigma_{22}$} \\
  \cmidrule(lr){5-6}
\cmidrule(lr){7-8}
\hline
$E[N_{obs}(\mathcal{D})]$ & $\tau^2$ & $\Sigma$ & $\eta$ 
& M3 & M4 
& M3 & M4 \\
\hline
300 & 3 & 1 & 0.3 
& 0.018 (0.015) & \textbf{0.016} (0.012) 
& 0.015 (0.025) & \textbf{0.011} (0.013) \\
\hline
500 & 1 & 1 & 0.3 
& \textbf{0.022} (0.018) & 0.024 (0.032) 
& 0.041 (0.097) & \textbf{0.026} (0.062) \\
\hline
500 & 3 & 0 & 0.3 
& 0.007 (0.004) & 0.007 (0.005) 
& 0.005 (0.003) & \textbf{0.004} (0.003) \\
\hline
500 & 3 & 1 & 0.3 
& 0.013 (0.005) & 0.013 (0.006) 
& \textbf{0.016} (0.071) & 0.020 (0.123) \\
\hline
500 & 3 & 1 & 0.5 
& 0.016 (0.012) & \textbf{0.014} (0.009) 
& 0.010 (0.018) & \textbf{0.009} (0.013) \\
\hline
1000 & 3 & 1 & 0.3 
& 0.013 (0.003) & \textbf{0.012} (0.004) 
& 0.007 (0.003) & \textbf{0.006} (0.002) \\
\hline
\end{tabular}
\label{Sigma_combined}
\end{table}

\begin{table}
\caption{$\Sigma_{11}$ and $\Sigma_{22}$ 90\% credible interval empirical coverage rates under models M3 and M4. For $\Sigma$: 0 corresponds to $(\Sigma_{11},\Sigma_{22}=(0.005,0.0025)$, and 1 corresponds to $(\Sigma_{11},\Sigma_{22})=(0.010,0.005)$.}
\centering
\begin{tabular}{cccccccc}
\hline
 &  &  &  & \multicolumn{2}{c}{$\Sigma_{11}$} & \multicolumn{2}{c}{$\Sigma_{22}$} \\
   \cmidrule(lr){5-6}
\cmidrule(lr){7-8}
\hline
$E[N_{obs}(\mathcal{D})]$ & $\tau^2$ & $\Sigma$ & $\eta$ 
& M3 & M4 
& M3 & M4 \\
\hline
300 & 3 & 1 & 0.3 
& 0.7 & \textbf{0.76} 
& 0.63 & \textbf{0.74} \\
\hline
500 & 1 & 1 & 0.3 
& 0.62 & \textbf{0.74} 
& 0.49 & \textbf{0.75} \\
\hline
500 & 3 & 0 & 0.3 
& 0.77 & \textbf{0.82} 
& 0.73 & \textbf{0.79} \\
\hline
500 & 3 & 1 & 0.3 
& 0.80 & \textbf{0.85} 
& 0.77 & \textbf{0.83} \\
\hline
500 & 3 & 1 & 0.5 
& 0.76 & \textbf{0.84} 
& 0.69 & \textbf{0.77} \\
\hline
1000 & 3 & 1 & 0.3 
& 0.78 & \textbf{0.83} 
& 0.79 & \textbf{0.83} \\
\hline
\end{tabular}
\label{Sigma_cov}
\end{table}

\section{Real data analysis}

The data are comprised of localized whale calls from Cape Cod Bay, MA on 2009-02-19. Specifically, the data arise as event times recorded on an array of 10 hydrophones that then require either automated or manual processing to a) associate the calls as coming from a single whale and then b) convert the associated group of times to a localization \citep{watkinsSoundSourceLocation1972}. These localizations  yield a spatial point pattern aggregated over the 24 hour period. (See Yack et al. [\textit{In prep}] for details on the methods.) 

As previously mentioned, the detection function considered in this passive acoustic setting is more complex than the exponential model assumed in our simulation studies. Namely, we have ten hydrophones in the bay and at least three of these must record the same call in order for the call to be spatially localized. Additionally, detection depends on the ambient noise at the hydrophone where louder ambient noise reduces the probability of detection \citep{thodeRoaringRepetitionHow2020,palmer2022accounting}. Therefore, let $\textbf{h}_k$ denote the location of the $k$th hydrophone and let $q_k(\bs)$ denote the probability of that hydrophone recording a call that originated by a whale located at $\bs$. Then, the probability of at least three hydrophones recording the call is
\begin{equation}
    p(\bs)=1-r_2(\bs)-r_1(\bs)-r_0(\bs)
\end{equation}
where
\begin{equation}
    r_0(\bs)=\prod_{k=1}^K(1-q_k(\bs))
\end{equation}
\begin{equation}
    r_1(\bs)=\sum_{k=1}^Kq_k(\bs)\prod_{j\neq k}(1-q_j(\bs))
\end{equation}
\begin{equation}
    r_2(\bs)=\sum_{k=1}^{K-1}\sum_{j=k+1}^K q_k(\bs)q_j(\bs)\prod_{l\neq j,k} (1-q_l(\bs))
\end{equation}
For our purposes, we opt to use exponential functions to model the individual $q_k(\bs)=\exp\{-d(\bs,\textbf{h}_k)/\eta\}$ where $d(\bs,\textbf{h}_k)$ is the Euclidean distance in kilometers from the call's origin location to the $k$th hydrophone. It is challenging to show algebraically whether or not $\eta$ is identifiable in this case. However, we can see that it is not in the limit. Taking $\eta\rightarrow 0$, it is clear that $q_k(\bs)\rightarrow 1$ for all $k$ and $\bs$. Conversely, as $\eta\rightarrow\infty$, $q_k(\bs)\rightarrow 0$ and thus, $p(\bs)\rightarrow 0$. Given this, for small and large $\eta$, $p(\bs)\approx1$ and $p(\bs)\approx0$ respectively. As a result, the detection surfaces are essentially flat, yielding an unidentifiable model in practice. To alleviate this issue and let $\eta$ dictate the decay rate rather than the overall detection probability, we define $r(\bs)=1-r_2(\bs)-r_1(\bs)-r_0(\bs)$ and
\begin{equation}
    p(\bs)=p^*\frac{r(\bs)}{\max_{\bs\in\mathcal{D}}r(\bs)}
\end{equation}
where $p^*$ is fixed and known. In this way, regardless of the value of $\eta$, there will exist some $\bs_0$ such that $p(\bs_0)=p^*$, which makes identifiability more feasible. In the distance sampling setting, this is analogous to $p(\bs)=1$ at the transect.

For the basis functions $\psi_k(\bs)$, we again take the first 100 eigenfunctions of an exponential covariance matrix with fixed range parameter $\phi=5$ (in kilometers). We discretize the region into 1,854 grid cells for computational purposes such that each grid cell is 1 km$^2$. For $\Sigma$, each diagonal component is given a uniform prior on the interval (0.5,50), of which the lower bound is chosen to be equal to half the length of one side of a grid cell and the upper bound is chosen such that the standard deviation is approximately half the distance between hydrophones. We run the MCMC algorithm for 100,000 iterations, discarding the first 20,000 as burn-in and using the remaining 80,000 iterations for posterior inference.

\subsection{Proof of concept}

As a proof of concept, we simulate from the model under these conditions, particularly with respect to the hydrophone-based detection function. We simulate $\beta_1,\dots,\beta_K\sim N(0,\tau^2)$ and then scale the resulting intensity function such that $E[N_{obs}(\mathcal{D})]=500$ in one simulation (Setting I) and $E[N_{obs}(\mathcal{D})]=1,500$ in the second simulation (Setting II) to explore the effects of varying the number of observed points. The model is evaluated based on recovery of $\eta$, $\Sigma$ and the true unobserved counts, $N(A)$, over various regions $A$. In particular, we look at counts in the whole domain $\mathcal D$, and two subregions $A_1$ (lower left rectangle pictured in Figure \ref{whalesim_lambda}) and $A_2$ (upper right rectangle pictured in Figure \ref{whalesim_lambda}). For prior distributions, we set $\eta\sim U(5,15)$ and $\Sigma_{ii}\sim U(0.5,50)$ for $i=1,2$.

Figure \ref{whalesim_lambda} shows the true log intensity surface along with the posterior mean estimates given each observed point pattern. As we can see, we recover the true intensity surface of $\log\lambda(\bs)$ under both data settings. Figure \ref{whalesim_sd} depicts the posterior standard deviation of $\lambda(\bs)$ at all pixels under both settings. From this, we can see that the pixel-wise uncertainty decreases substantially as the number of points observed increases.

\begin{figure}
    \centering
    \includegraphics[width=0.8\linewidth]{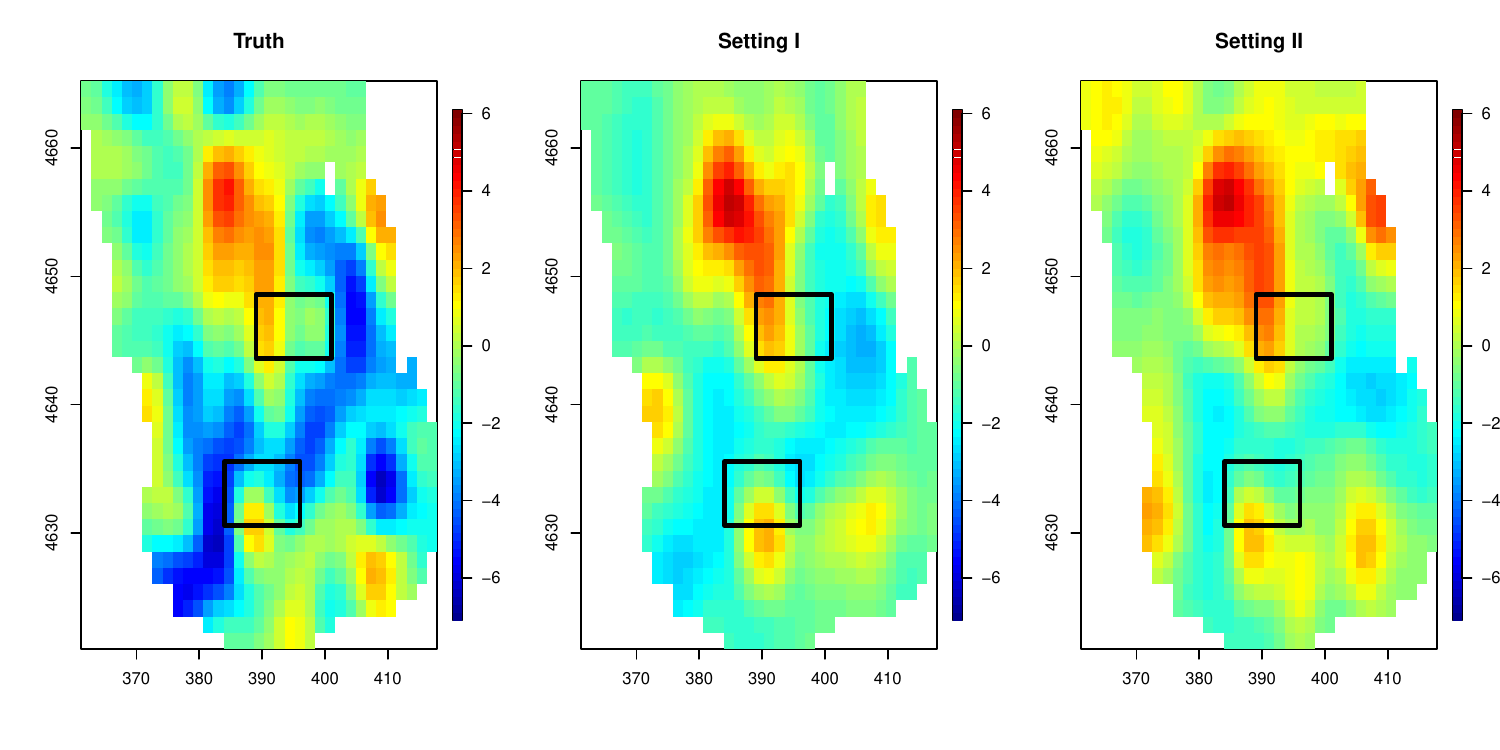}
    \caption{Maps of $\log\lambda(\bs)$ for the truth (left), posterior mean given the point pattern with $N_{obs}(\mathcal{D})=485$ (middle), and posterior mean given the point pattern with $N_{obs}(\mathcal{D})=1514$ (right). The black borders indicate the two subregions within which we will predict actual counts. Note that for the intensity of $N_{obs}(\mathcal{D})=485$, we have scaled it up such that it is comparable to the other two intensities.}
    \label{whalesim_lambda}
\end{figure}

\begin{figure}
    \centering
    \includegraphics[width=0.6\linewidth]{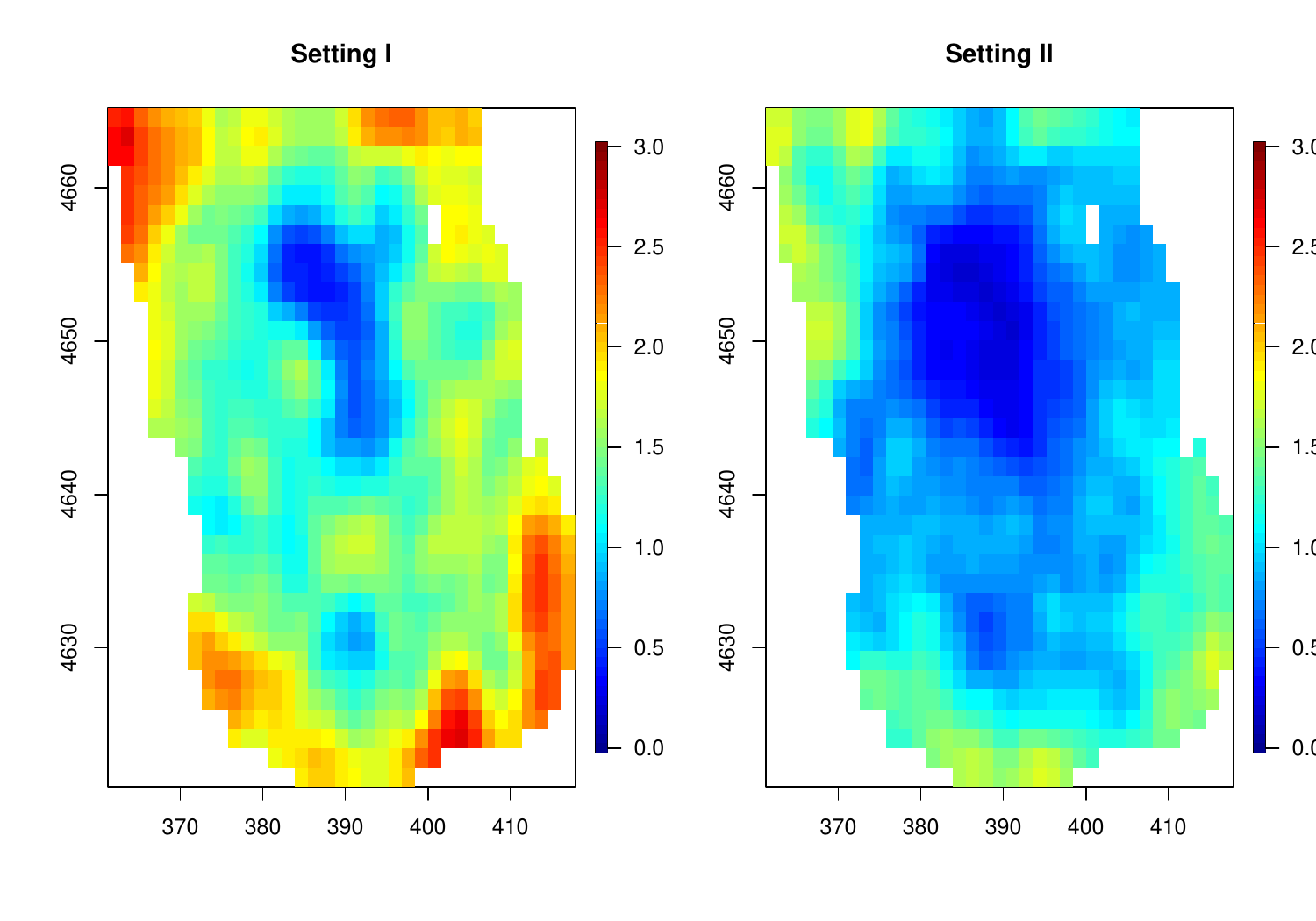}
    \caption{Pixel-wise posterior standard deviations for $\log\lambda(\bs)$ for the point patterns with $N_{obs}(\mathcal{D})=485$ (left) and $N_{obs}(\mathcal{D})=1514$ (right).}
    \label{whalesim_sd}
\end{figure}

Table \ref{whalesim_params} reports posterior medians, standard deviations, and 90\% credible intervals for all model parameters. The model recovers all parameters well, with the exception of $\Sigma_{11}$ being overestimated. Furthermore, Table \ref{whalesim_counts} gives estimates of $N(A)$, the true latent number of points. Our model predicts true counts in each subregion, $A_1$ and $A_2$, and the domain, $\mathcal{D}$. When integrating over $\mathcal{D}$, the predicted counts are large and the posterior distributions are heavy tailed, resulting in wide intervals. Importantly, the true counts are not the observed counts, but rather the latent counts before thinning and displacement. Thus, recovery of the true value relies on recovery of the intensity function and both degradation components. As we can see, the posterior median tends to slightly underestimate the the true count. This can be attributed to an overestimation of $\eta$ corresponding to higher estimated $p(\bs)$ across the region. The 90\% credible intervals cover the truth in all settings.

\begin{table}
\caption{Proof of concept degradation parameter posterior median (SD) and 90\% credible interval.}
\centering
\begin{tabular}[t]{cccccc}
\hline
 & & \multicolumn{2}{c}{$N_{obs}(\mathcal{D})=485$} & \multicolumn{2}{c}{$N_{obs}(\mathcal{D})=1,514$} \\
   \cmidrule(lr){3-4}
\cmidrule(lr){5-6}
\hline
 & Truth & Median (SD) & 90\% CI & Median (SD) & 90\% CI \\
\hline
$\eta$ & 10 & 10.82 (2.39) & (6.57,14.49) & 10.41 (1.96) & (7.39,13.93)\\
$\Sigma_{11}$ & 3 & 4.13 (1.68) & (1.61,7,19) & 2.00 (0.65) & (1.15,3.27) \\
$\Sigma_{22}$ & 5 & 7.54 (2.33) & (3.98,11.61) & 5.87 (1.09) & (4.02,7.63)\\
\hline
\end{tabular}
\label{whalesim_params}
\end{table}

\begin{table}
\caption{Proof of concept $N(A)$ posterior median (SD) and 90\% credible interval.}
\centering
\begin{tabular}[t]{ccccccc}
\hline
 &  \multicolumn{3}{c}{$N_{obs}(\mathcal{D})=485$} & \multicolumn{3}{c}{$N_{obs}(\mathcal{D})=1,514$} \\
   \cmidrule(lr){2-4}
\cmidrule(lr){5-7}
\hline
 &  Truth & Median (SD) & 90\% CI & Truth & Median (SD) & 90\% CI \\
\hline
$\mathcal{D}$ & 2,860 & 2,467 (2,475) & (1,881, 8,831) & 8,662 & 8,420 (4,727) & (6,197, 20,491)\\
$A_1$ & 37 & 35 (35) & (12, 117) & 116 & 90 (44) & (49, 187) \\
$A_2$ & 173 & 122 (35) & (73, 189) & 519 & 483 (66) & (384, 601)\\
\hline
\end{tabular}
\label{whalesim_counts}
\end{table}

\subsection{Analysis}

Lastly, we apply our model to the real data. In selecting $p^*$, we rely on auxiliary information. We use the individual hydrophone detection functions estimated in \cite{palmer2022accounting}, which depend on ambient noise at the receiver. On the day of observation, ambient noise ranged between 103-106 dB, so we approximate the curve used for an ambient noise level of 104.5. We use these results only to estimate the maximum detection probability in the region, which comes out to $p(\bs)\approx1$ inside of the hydrophone array. \cite{palmer2022accounting} used a different call processing approach from the data we have. Because their approach detects roughly four times the number of calls, we set $p^*=0.25$. In this sense, $p^*$ serves as a calibration factor accounting for differences in call processing and
detector sensitivity between the locations in \cite{palmer2022accounting} and the present dataset. We note that this adjustment assumes processing differences act approximately as a global multiplicative scaling of detection probability. For a prior, we assume $\eta\sim Unif(5,15)$.

Figure \ref{realdata_fig} illustrates the observed point pattern and posterior means for $\log\nu(\bs)$, $\log(p(\bs)\lambda(\bs))$, and $\log\lambda(\bs)$. We see how $\nu(\bs)$ captures the intensity of the observed data, and how each degradation component contributes to yielding different intensities--ultimately arriving at the latent intensity function $\lambda(s)$. Posterior summaries of both the degradation parameters and the expected counts can be found in Table \ref{realwhale_tab}. The posterior of $\eta$ is roughly centered within the uniform interval and the 90\% CI length indicates some learning relative to the prior $\eta\sim U(5,15)$. The displacement variances, $\Sigma_{11}$ and $\Sigma_{22}$, have posterior median estimates of 1.32 and 5.28 indicating greater displacement error in the vertical direction. We estimate a total of 4,021 calls over the course of the 24 hour period within $\mathcal{D}$. North Atlantic right whales call on average 0-200 times per hour, and this rate is highly dependent on life history stage, habitat, and behavior \citep{matthewsOverviewNorthAtlantic2021}. Assuming an average of 6 calls per whale per hour, this corresponds to an estimated $\sim$28 vocalizing whales present in Cape Cod Bay on 2009-02-19. This is in line with estimates of abundance \citep{garciaAcousticAbundanceEstimation2025}, suggesting appropriate calibration of our model. 

\begin{figure}
    \centering
    \includegraphics[width=0.8\linewidth]{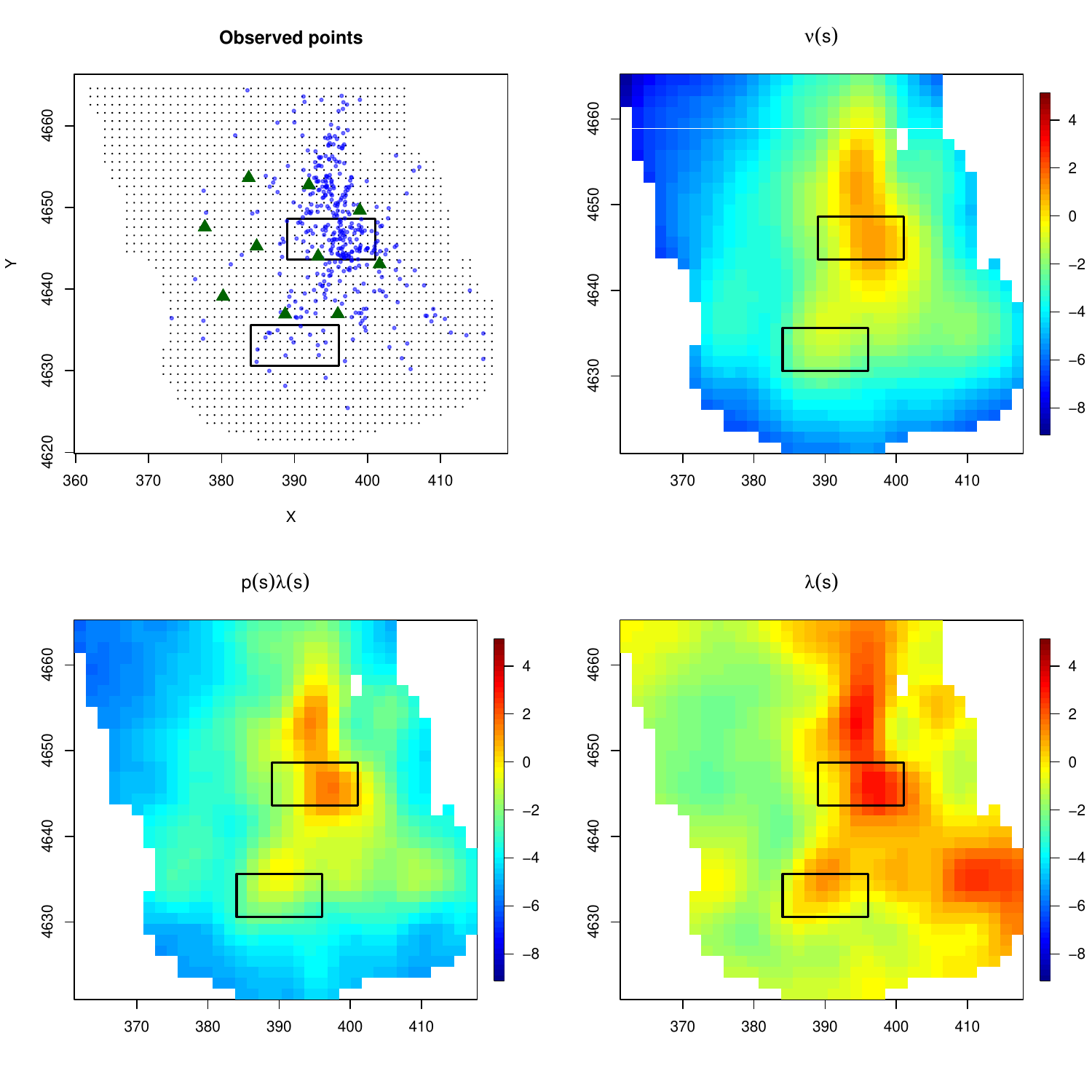}
    \caption{Observed point pattern (blue dots) and hydrophone locations (green triangles) (top left), and posterior means for log $\nu(s)$ (top right), $p(s)\lambda(s)$ (bottom left), and $\lambda(s)$ (bottom right).}
    \label{realdata_fig}
\end{figure}

\begin{table}
\caption{Posterior summaries for degradation parameters and $N(A)$ for given regions}
\centering
\begin{tabular}[t]{cccccc}
\hline
 \multicolumn{3}{c}{Degradation} & \multicolumn{3}{c}{Counts} \\
   \cmidrule(lr){1-3}
\cmidrule(lr){4-6}
\hline
 & Median (SD) & 90\% CI & & Median (SD) & 90\% CI \\
\hline
$\eta$ & 8.81 (1.78) & (6.41,12.29) & $\mathcal{D}$ & 4,021 (6,326) & (2,224, 14,650) \\
$\Sigma_{11}$  & 1.32 (0.79) & (0.58,2.02) & $A_1$ & 113 (55) & (61, 225) \\
$\Sigma_{22}$ & 5.28 (3.54) & (2.02,13.55) & $A_2$ & 606 (78) & (488, 746)\\
\hline
\end{tabular}
\label{realwhale_tab}
\end{table}

\section{Summary and future work}

We have investigated a challenging, and relatively unexplored, modeling problem for degraded point pattern data where degradation involves both thinning (missed detection) and displacement (location error). In the context of Poisson processes, we proposed a  hierarchical model to analyze such data, arguing that thinning followed by displacement is a natural hierarchical ordering. We developed formal/theoretical arguments for when our model will be able to identify the degradation mechanisms as well as the true Poisson process intensity. In this regard, we contributed a novel theoretical result pertaining to Gaussian deconvolution. In addition to theoretical arguments, we demonstrated empirically through simulation study that we are able to recover model parameters and true counts under this framework.

Furthermore, we applied this modeling framework to a challenging analysis of localized North Atlantic right whale calls in Cape Cod Bay. Our data consists of locations estimated from the time difference of arrival of whale calls \citep{watkinsSoundSourceLocation1972} recorded on an array of hydrophones over a window of time. This form of passive acoustic monitoring necessarily misses some calls since the probability of detection decreases with distance from the hydrophone. Furthermore, none of the ``observed"  locations in our point pattern are exact.  Since they are estimated, they have intrinsic error which manifests as displacement. This localization step emphasizes the need for the proposed degraded modeling, as the spatial point pattern is not observed directly. The thinning mechanism on the real data is more complicated than the single-transect distance sampling functions discussed earlier, so we employed a simple simulation as a proof of concept before analyzing the real data.

There are multiple natural extensions to our work. One could consider the addition of ``ghost" points, i.e., false positives in the point pattern. In our real data setting, this could translate to the calls of humpback whales being falsely labeled as right whales. Extensions to the dynamic setting where we have a spatial intensity function changing over time could be of interest in a number of applications. For example, knowing where right whales are likely to be at fine temporal and spatial resolutions is of the highest management importance in both US and Canadian waters. In addition, since the spatial intensity represents whale calling behavior, this provides a framework to quantitatively assess how anthropogenic noise might impact the acoustic environment for different species. With respect to our real data analysis, we acknowledge that our degraded model is a simplification of the real process. That is, the probability of detecting a call is a function of the noise level at the hydrophone, the source level of the whale call itself, and transmission loss across the bay \citep{palmer2022accounting}. Furthermore, localization error is known to increase outside of the hydrophone array \citep{urazghildiiev2013comparative}. A natural extension of our model would be a spatially-dependent error variance to account for this variation.

\section{Acknowledgements}

Funding for this work was provided in part by grants N000142312562 and N000142412501 from the US Office of Naval Research. 

\printcredits

\bibliographystyle{elsarticle-harv}

\bibliography{cas-refs}



\end{document}